\documentclass{article}

\usepackage{graphicx}%
\usepackage{multirow}%
\usepackage{amsmath,amssymb,amsfonts}%
\usepackage{amsthm}%
\usepackage{mathrsfs}%
\usepackage{dsfont}
\usepackage[title]{appendix}%
\usepackage{xcolor}%
\usepackage{textcomp}%
\usepackage{manyfoot}%
\usepackage{booktabs}%
\usepackage{algorithm}%
\usepackage{algorithmicx}%
\usepackage{algpseudocode}%
\usepackage{listings}%
\usepackage{subcaption}
\usepackage{threeparttable} 
\usepackage{csquotes}

\usepackage{authblk} 

\usepackage{natbib}
\usepackage[colorlinks=true, allcolors=blue]{hyperref}

\theoremstyle{thmstyleone}%
\newtheorem{theorem}{Theorem}[section]
\newtheorem{definition}{Definition}
\theoremstyle{thmstyletwo}%

\theoremstyle{thmstylethree}%

\graphicspath{{figures/}}

\begin{document}

\title{Vector Traits Shape Disease Persistence: A Predator–Prey Approach to Dengue}

\author[1,2]{Piyumi Chathurangika}
\author[1]{Tharushika Peiris}
\author[3]{Lakmini S. Premadasa}
\author[1]{S. S. N. Perera}
\author[4]{Kushani De Silva\thanks{Email: sarachchilag@lamar.edu, kushanipdesilva@gmail.com}}

\affil[1]{Research \& Development Centre for Mathematical Modeling, Department of Mathematics, Faculty of Science, University of Colombo, Colombo 00030, Sri Lanka}
\affil[2]{Department of Electrical and Electronics Technology, Faculty of Technology, Rajarata University of Sri Lanka, Mihintale 50300, Sri Lanka}
\affil[3]{Texas Biomedical Research Institute, San Antonio, Texas 78227, USA}
\affil[4]{Department of Mathematics, Lamar University, Beaumont, Texas 77705, USA}

\maketitle

%
%

\date{}


\begin{abstract}
Dengue continues to pose a major global threat, infecting nearly 390 million people annually. Recognizing the pivotal role of vector competence ($v_c$), recent research focuses on mosquito parameters to inform transmission modeling and vector control strategies.This study models interactions between Aedes vectors and dengue pathogens, highlighting $v_c$ as a key driver of within-vector infection dynamics and endemic persistence. Using a predator–prey framework, we show that endemic conditions emerge naturally from the biological interplay between the vector’s strategies to pathogen pressure and we prove global stability of such conditions. Our results reveal that under tropical and subtropical environmental pressures, the innate immune system of vectors cannot offset high $v_c$ during endemic outbreaks, highlighting a fundamental biological trade-off: vectors can evolve increased transmission potential but cannot enhance immune capacity. This constraint defines the limits of their evolutionary response to pathogen-driven selection and drives instability in disease transmission dynamics.
\end{abstract}



\maketitle

\section{Introduction}\label{sec1}

Amid enduring global health challenges that have long threatened human well-being, dengue continues to persist as a major concern, causing nearly 390 million infections annually and remaining without a definitive cure or an effective vaccine despite decades of scientific effort \citep{bancroft1906aetiology,ulgheri2025decoding}. Recent decades have witnessed a remarkable expansion of dengue, with roughly one-quarter of the global population now living in regions where the disease is endemic, underscoring its enduring and widespread public health burden \citep{couderc2025decoding}. For decades, dengue research has focused on its clinical and epidemiological aspects—work that has provided essential insights into the disease. However, competent mosquitoes are undoubtedly the most important puzzle piece in the disease transmission. Recognizing this importance, recent studies now emphasize the ecological and social contexts of \textit{Aedes} to advance innovative strategies for modeling dynamics, vector control, and long-term mitigation. Genetic alterations of mosquitoes to target key parameters influencing their vectorial capacity—such as insecticide resistance, and sensitivity to environmental factors like weather and climate—has emerged as a promising direction in vector control research \citep{Vs42,Vs33,Vs31}. The effectiveness of the widely adopted Wolbachia-infected mosquito strategy to block pathogen transmisson is largely determined by its impact on key parameters of vectorial capacity—particularly the mosquito’s ability to transmit the pathogen, known as vector competence ($v_c$) \citep{wang2025gene,Vs33,Vs44,souza2019aedes}. Transgenesis approach is another method of vector control to curtail disease transmission by impairing $v_c$ \citep{Vs45,Vs46,merkling2025dengue}. Among these, several studies have specifically examined $v_c$ in the framework of pathogen-vector interactions offering within-vector infection dynamics \citep{couderc2025decoding,Vs_12, Vs_13, Vs_39,Vs_8,G_41}. Another potential mechanism for disease control via the vector—demonstrated experimentally—is through the Toll pathway. Inactivating this key immune signaling pathway in mosquitoes, for example by silencing components like MYD88, reduces their ability to limit or inactivate viral infections, resulting in increased viral replication. Quantifying the virus inactivation rate mediated by the Toll pathway provides critical insight into the strength of vector immune defenses and their role in controlling virus levels within the vector \citep{mukherjee2019mosquito,tauszig2002drosophila,xi2008aedes}. Given the demonstrated significance of key parameters—particularly $v_c$ and pathogen inactivation rate showing immune strength—it is essential to incorporate them when investigating the transmission dynamics of mosquito-borne diseases.\\


The intrinsic imbalance between opposing forces of gain and loss fundamentally governs fluctuations in a species’ population abundance. A stable equilibrium of such fluctuations will provide a sustainable ecosystem. When the forces of gain outweigh those of loss, the population experiences exponential growth; conversely, when losses surpass gains, the population faces inevitable decline toward extinction. This fundamental asymmetry highlights the fragile equilibrium that dictates whether a population persists or collapses. Therefore, the ways in which species interact, shaped by their biological traits, determine whether the relationship between two populations fosters a stable persisting ecosystem. In the context of dengue, the introduction of the pathogen into the vector’s body through a blood meal imposes selective pressure on the mosquito’s immune system, leading to corresponding gains or losses for the vector. The outcome of these gains or losses is dictated by the strategies the vector reasons to counter, assuming the pathogen are equally guided by the reason. A strategy is a set of rules that an individual (or player) follows to make decisions in response to the actions of others, with the goal of maximizing its fitness, payoff, or success in a given context \cite{eG_03}. In analogy to a biological game, we simplify the system’s complexities by assuming that each player (vector and pathogen) has exactly two strategies from which to choose. The vector’s two strategies are determined by how it allocates its energetic and molecular resources between its primary reproductive functions (tolerance), such as vitellogenesis and egg maturation \citep{Im_15, Im_17}, and the activation of its immune system, RNA interference (RNAi) pathway, Toll and IMD signaling, and melanization. to combat the pathogen (resistance) \citep{Im_01, G_44}. However, activating these immune pathways comes at a cost, both metabolically and in terms of delayed or reduced fecundity, highlighting a fundamental trade-off between resistance and reproduction \citep{Im_18}. Based on these stategies, pathogen will either replicate (upon tolerance) or silenced (upon resistence). This reflects the structure of classical predator–prey dynamics, in which the predator’s (pathogen’s) growth—its gains or losses—depends on the abundance of the prey (vector), and its biological traits.\\

Predator–prey models have long served as a foundational framework for understanding the dynamics of interacting populations. Over the years, these models have been refined to include more realistic assumptions such as functional responses, predator handling time, and environmental variability \citep{Ho_02, Ho_03}. A comprehensive review of various predator–prey modeling approaches highliting various functional responses is provided by \citep{novak2021geometric}. These models have been effectively applied to explore disease dynamics involving vertebrate hosts and their pathogens, as demonstrated in \citep{fenton2010applying, agyingi2020modeling}, with an extensive review available in \citep{friedman2022hierarchy}. Moreover, predator–prey frameworks have been extensively applied to the study of parasite–host and insect–host interactions, providing insights into how natural enemies of parasites can be leveraged for biological pest control. Notably, these systems are often modeled in discrete time, reflecting the generational or seasonal structure of interactions, which distinguishes them from classical continuous-time Lotka–Volterra models \cite{PP_100, hassell2000host, hassell1986generalist, ives2005synthesis, PP_13, G_63,PP_05}. Furthermore, several studies have conceptualized the pathogen as the prey and the host immune system as the predator, framing infection dynamics within a predator–prey paradigm \citep{PP_12, G_61, G_45}. Unlike previous predator–prey applications, this study introduces a novel inversion: the Aedes mosquito vector, the transmitter of dengue, is modeled as the prey, while the dengue virus is treated as the predator, owing to its enigmatic nature that continually challenges scientific understanding. This predator–prey framework demonstrates how the vector’s biological traits - specifically $v_c$ - to capture how the biological characteristics of the vector mediate the impact of pathogen pressure on mosquito populations, ultimately shaping disease transmission dynamics.\\

The paper is structured as follows: Section 2 provides a detailed explanation of the development of the mathematical model, grounded in biological context. Section 3 presents the results of the model analysis, including equilibria, their stability, and supporting numerical findings. Finally, Section 4 concludes the paper by summarizing key findings, limitations, and suggesting directions for future research.

\section{The model formulation}

Dengue transmission dynamics are commonly modeled using SIR–SI compartmental frameworks (see Fig. \ref{fig:fig-a}). In such models, the human population is divided into three mutually exclusive compartments—susceptible ($S_h$), infected ($I_h$), and recovered ($R_h$)—while the vector population is divided into two compartments—susceptible ($S_v$) and infected ($I_v$). External (e.g., climate, mobility) and internal (e.g., immunity, pathogen evolution) factors drive continuous movements of individuals within and between these compartments. The dengue virus (pathogen) acts as the primary ecological driver of these transitions. Conceptually, the traditional SIR–SI model can be simplified by introducing an explicit pathogen compartment, replacing the abstract “environmental pressure” that drives transitions between vector states (see Fig. \ref{fig:fig-ab}(b)). Within the vector population, the transition from susceptible ($S_v$) to infected ($I_v$) occurs through pathogen predation, wherein the virus functions ecologically as a predator that exploits susceptible vectors as its resource base. This interaction mirrors a predator–prey dynamic, in which the pathogen’s success depends on its ability to “capture” and utilize susceptible vectors, thereby sustaining its persistence and propagation within the host–vector ecosystem.\\

One of the fundamental limitations of assuming continuous compartmental transitions driven by constant transmission rates in the SIR–SI or predator–prey framework is that it overlooks the influence of additional ecological and epidemiological factors. Variables such as mosquito population density, human mobility and urbanization, climatic conditions (e.g., temperature and rainfall), and prior immunity levels within the human population profoundly shape these transitions. The combined effects of these fluctuating factors often disrupt equilibrium conditions, creating imbalances that trigger outbreaks. As a result, dengue transmission exhibits episodic dynamics—characterized by intermittent surges and declines in infection—rather than a smooth and continuous flow of disease transmission. This assumption holds only under endemic conditions, where multiple serotypes co-circulate and sustained human–pathogen–vector contact is maintained. Such conditions are rarely stable in real-world dengue dynamics, although recent studies suggest that many regions may indeed experience endemic transmission. Under endemic assumptions, population abundances are expected to exhibit continuous fluctuations, driven by the imbalanced forces between the “prey” (susceptible vectors) and the “predator” (pathogen). Despite uncertainties regarding endemicity, this framework enables the identification of conditions under which endemic transmission is sustained, particularly in relation to vector traits and vector biological parameters.

\begin{figure}[htb]
	\centering
		\begin{subfigure}{0.45\linewidth}
			\includegraphics[width=\textwidth]{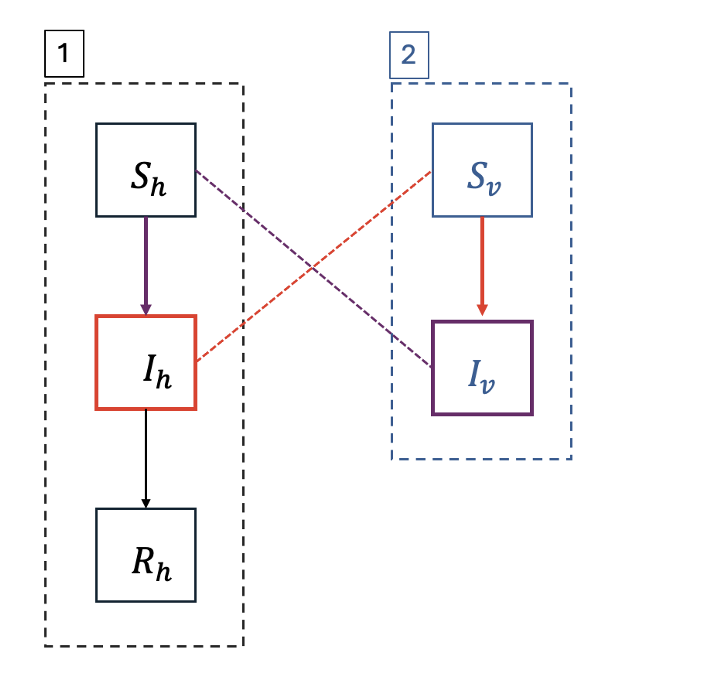}
			\caption{\label{fig:fig-aa}}
		\end{subfigure}
		\begin{subfigure}{0.45\linewidth}
			\includegraphics[width=\textwidth]{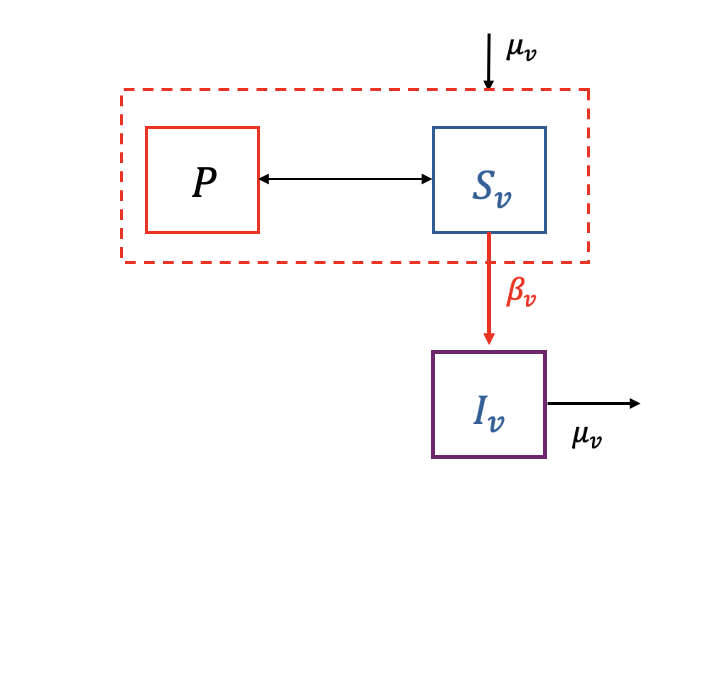}
			\caption{\label{fig:fig-ab}}
		\end{subfigure}
      \caption{The transmission structures: (a) Schematic diagram of the coupled SIR (block 1)–SI (block 2) system. Vector transition from susceptibility to infection is driven by $I_h$ (red dashed line), while human transition from $S_h$ to $I_h$ is driven by $I_v$ (purple dashed line). (b) The infected human population drives vector–pathogen interactions, shifting vectors from susceptible to infected. The solid double-sided arrow denotes the predation relationship between the pathogen and susceptible vectors.}
\label{fig:fig-a}
\end{figure}

A generalized predator–prey framework describing vector–pathogen interactions is formulated in Eq. \eqref{general}. Here, 
$S_v$ denotes the susceptible vector population density (prey), while $P$ represents the pathogen population density (predator), capturing the antagonistic dynamics that underpin infection-driven ecological regulation. Population densities represent the number of individuals per unit area, analogous to those in classical predator–prey systems. The functions $ g(S_v) $ and $ f(S_v) $ respectively denote the prey growth rate and the consumption rate of the prey, while $ \delta $ represents the predator's removal rate from the system, i.e. vector immunity inactivating the pathogen. Here, the pathogen’s consumption rate is assumed to contribute equally to the predator’s gain and the prey’s loss. With the general model in ~\eqref{general}, we justify in the next subsections the functional forms for $ g(S_v) $ and $ f(S_v) $ based on vector biology.

\begin{equation}
\begin{cases}
\vspace{0.1cm}
\dfrac{dS_v}{dt}=g(S_v)-Pf(S_v)\\
\dfrac{dP}{dt}=P f(S_v)-\delta P
\end{cases}
\label{general}
\end{equation}

\subsection{Vector growth rate $ g(S_v) $}

Immune responses such as activation of the RNA interference (RNAi) pathway or Toll and IMD signaling require substantial energy and molecular resources. These same resources are also essential for reproductive processes like vitellogenesis, egg maturation, and oviposition. As a result, during pathogen pressure, mosquitoes often divert its energy from reproduction toward immunity, leading to reduced fecundity, delayed oviposition, or impaired egg viability. This trade-off has been observed across multiple insect systems and is supported by studies showing that dengue-infected \textit{A. aegypti} exhibit altered metabolic and immune signaling that can suppress reproductive investment. However, the extent of this inverse relationship can vary depending on factors such as nutritional status, genetic background, environmental conditions, and the severity of infection. In some cases, mosquitoes may prioritize reproduction even when infected, a strategy known as terminal investment \citep{F_12, Im_18, G_64}, especially in high-risk environments. Nonetheless, under typical conditions, the inverse relationship between immune activation (resistance) and reproduction (tolerance) is a key factor influencing vector competence ($v_c$) and transmission potential.


\begin{figure}[ht!]
		\centering
				\includegraphics[width=0.5\textwidth]{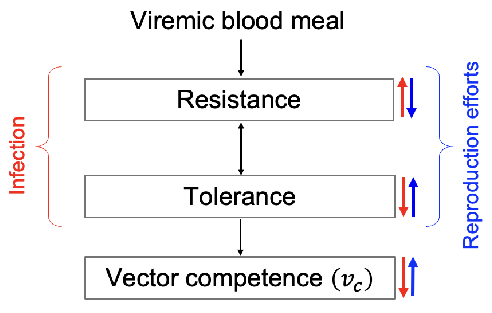}
		\caption{\label{trade-off} Fitness trade-off between resistance (immune defense) and tolerance (reproduction) within the vector's body after taking a viremic blood meal. Red arrows represent the flow of trade-off when resistance is increased, whereas blue arrows show the same flow when the tolerance is increased.}
	\end{figure}

In the absence of the pathogen, vectors ($S_v$) will grow according to the self-limiting logistic process:
\begin{align*}
\mu_vS_v\left( 1-\dfrac{S_v}{K}\right), 
\end{align*}
 where $\mu_v$ is the growth rate of $S_v$ and $K$ is the carrying capacity prior to limitation by disease transmission. When reproductive fitness increases due to higher tolerance, resistance is reduced, permitting a higher $v_c$. In this case, the carrying capacity scales relative to $v_c$, i.e. $v_cK$. Conversely, when reproduction is constrained by higher resistance, lower $v_c$ values proportionally reduce the carrying capacity. Vectors have a natural birth rate, $\mu_v$, but under pathogen pressure, reproduction is suppressed by resistance efforts. Given the critical relationship between vector reproduction and $v_c$ (Fig. \ref{trade-off}), the birth rate is scaled by $v_c$ to reflect pathogen-induced stress, yielding an updated growth rate of $v_c\mu_v$. Thus, the vector growth rate $g(S_v)$ can be expressed as,
\begin{equation}
	\label{g(v)}
	g(S_v)=v_c \mu_v S_v\left(1-\dfrac{S_v}{v_c K}\right).
\end{equation}
\subsection{Consumption rate $f(S_v)$}

Despite extensive study of predator–prey dynamics in ecology, experimental quantification of how pathogens exploit their vectors remains virtually unexplored. In this context, dengue virus–\textit{Aedes} interactions present a unique system in which the pathogen can be conceptualized as a predator acting on its vector ``prey.” This study represents the first effort to formalize this relationship using a predator–prey framework, thereby bridging classical ecological theory with vector-borne disease dynamics. To establish a mechanistic baseline, we systematically investigate all three Holling-type functional responses (Type I, II, and III), which describe different modes of predator consumption. With this, we aim to determine whether using one type of consumption rate over the others provides a clear advantage in understanding conditions for disease persistence. The general form of Holling's type functional response is given in Eq.~\eqref{holling-general}.
\begin{align} \label{holling-general}
	f_q(S_v) = \begin{cases}
		aSv, & \text{ for } q=0,\\
		\dfrac{aS_v^q}{1+ahS_v^q}, & \text{for } q =1,2,
	\end{cases}
\end{align}
where $a$ is the attack rate and $h$ is the handling time. The values $q=0,1,2$ correspond to Type I, II, and III functional responses respectively. \\

When a pathogen enters a vector and establishes infection, multiple internal barriers must be overcome for replication. Once these barriers are cleared, the vector becomes infectious and can transmit the pathogen to another host, changing its status from susceptible to infected (Fig. \ref{fig:fig-ab}). The time required for this process is called the extrinsic incubation period, which, in predator–prey terminology, can be hypothesized to the handling time, $h$. On the other hand, the attack rate ($a$), representing the likelihood of infection, depends on two factors: (1) the probability that a vector facilitates pathogen transmission by biting a host, denoted $\beta_v$, and (2) the vector’s susceptibility to infection, denoted $v_c$. That is,
\begin{align*}
	p\left(\text{attack or infection} \right) &= p\left(\text{host bite} \right)p\left( \text{infection$|$host bite}\right), \\
	a&=\beta_v v_c. \notag 
\end{align*}
Empirical studies use metrics like Infection Rate (IR), Dissemination Efficiency (DE), and Transmission Efficiency (TE) to assess $v_c$, typically ranging from 0 to 1 \citep{Vs8, hardy1983, Vs14, Op_2}. In our study, $v_c$ is assumed to be assessed based on infected human population densities from successfully infected \textit{Aedes}. Therefore, in this study, $v_c$ is assumed to reflect TE. Below we showcase the final ODE system (Eqs. \eqref{holling-specific}- \eqref{final model}).
\begin{align} \label{holling-specific}
	f_q(S_v) = \begin{cases}
		\beta_v v_c S_v, & q=0,\\
		\dfrac{\beta_v v_c S_v^q}{1+\beta_v v_c E_p S_v^q}, & q=1,2. 
	\end{cases}
\end{align}
\begin{equation}
	\begin{cases}
		\vspace{0.1cm}
		\dfrac{dS_v}{dt}=v_c \mu_v S_v\left(1-\dfrac{S_v}{v_c K}\right)-f_q\left( S_v\right) P &\equiv S_v F\left(S_v,P \right) ,\\
		\dfrac{dP}{dt}= \alpha f_q\left( S_v\right) P-\delta P &\equiv P G\left(S_v,P \right),
	\end{cases}
	\label{final model}
\end{equation}
where $\alpha$ is the replication rate, scaling how efficiently the pathogens replicate inside the vector body. Based on biological considerations, the initial conditions are,
\begin{align} \label{Eq. init}
S_v\left( 0\right)>0, \text{ and } \, P\left(0 \right) >0.
\end{align}
The parameter definitions of the dynamical system in Eqs. \eqref{holling-specific} - \eqref{final model} are given in Table \ref{Table: parameter definition}. The functional response is influenced by the predator's maximum consumption capacity, search speed, and attack success \citep{Ho_03}. As a result, the consumption rate curves (Type I, II, or III) are a result of the trade-offs between these factors. In a Type I functional response, consumption occurs in direct proportion to prey density, assuming no limit to the predator’s consumption capacity—that is, consumption is independent of prey density. However, at the highest prey densities, the upper bound of consumption is given by $\beta_v v_c / (1 + \beta_v v_c E_p)$. This indicates that in Type II and Type III functional responses, which incorporate nonlinear dynamics, the maximum consumption capacity is always lower than that of the linear Type I response.
\begin{table}[ht!]
\footnotesize
	\caption{Parameter definitions of the model given by Eqs.~\eqref{holling-specific} - \eqref{final model}.}
	\centering
	\begin{tabular}{lll}
		\hline
		Parameter	 	& Definition 						&	Literature Value	\\
		\hline
		$\mu_v$ 		& Vector mortality rate				&	1/6 \citep{VDCI}\\
		$\beta_v$ 		& Transmission rate from host to vector	&	0.375 \citep{G_12}	\\
		$E_p$ 		& Extrinsic incubation period (weeks)	&	2 \citep{G_96}		\\
		$K$ 			& Carrying capacity of vectors			&	1			\\
		$v_c$ 		& Vector competence				&	-			\\
		$\alpha$		&	Pathogen replication rate			&	-	\\
		$\delta$ 		& Pathogen inactivation rate			&	-			\\
		\hline
	\end{tabular}
\label{Table: parameter definition} 
\end{table}

Type II functional responses provide a baseline for consumption that scales with prey density. Even with abundant pathogens, the consumption rate is limited by the host-bite probability ($\beta_v < 1$). In Type III, hyperbolic dependence on $S_v^2$ reflects vector competition for blood meals, reducing pathogen consumption capacity and highlighting how vector density and behavior constrain transmission dynamics. We can therefore assume a Type III functional response during periods of rapid mosquito population growth, such as following the rainy season. In regions with monsoon rains, increased breeding sites drive a delayed rise in mosquito abundance—for example, in Sri Lanka, this time lag is approximately 10 weeks \citep{OP_1}. Consequently, even if the transmission rate remains constant, $v_c$ may oscillate between Type II and Type III dynamics depending on seasonal rainfall (example work in \citep{chathurangika2024determining}), highlighting the critical influence of environmental conditions on pathogen–vector interactions.
\section{Equilibrium analysis}

In this section, we showcase the preliminary results, equilibrium analysis, and their stability of the system in \eqref{holling-specific}-\eqref{final model}. Firstly, Theorem \ref{th:posi} states the positivity of the solutions to the system \eqref{holling-specific}-\eqref{final model}.\\

\begin{theorem} \label{th:posi}
All solutions of the system in \eqref{holling-specific}-\eqref{final model} with initial conditions $S_v(0), P(0) >0$ remains positive for all $t>0$.
\end{theorem}
	\begin{proof}
	From the existence and uniqueness theorem, initial value problem of system in \eqref{holling-specific}-\eqref{final model} has a unique continuous solution. \\
	
	Suppose a trajectory with initial conditions $S_v(0), P(0) >0$ leaves the first quadrant of $S_v$, $P$ phase plane. Then it must be either $P(t_0)=0$ or $S_v(t_0)=0$ at some $t_0>0$ . Without loss of generality, let us assume $P(t_0)=0$. Then the trajectory must be on $S_v $ axis at $t_0$. Notice that anywhere on the $S_v$ axis, $P=0$ and thus $dP/dt=0$. Therefore the trajectory only depends on $dS_v/dt$. Hence, any trajectory arrive $S_v$ axis must remain on $S_v$ axis. Therefore such trajectory cannot cross the $S_v$ axis. This is a contradiction. Therefore, All solutions of the system with initial conditions $S_v(0), P(0) >0$ remains positive.
\end{proof}

\begin{definition} \label{def}
	Given the system in \eqref{holling-specific}-\eqref{final model}, we say a point $\left( S_v^*,P^*\right) $ an equilibrium of this system if $\dot{P}=0$ and $\dot{S_v}=0$ simultaneously. For such a point, the constant function $\left( S_v(t),P(t)\right) \equiv \left( S_v^*,P^* \right)$ is a solution of the system.
\end{definition}
\vspace*{3mm}
\noindent
According to Definition \ref{def}, for the three types of functional responses, the equilibrium points occur at the following instances:
\begin{description}
	\item[$E_1$ (Trivial):] $S_v=0, F\left( S_v,P\right) \neq 0$ and $P=0, G\left(S_v,P \right) \neq 0$,\\
	\item[$E_2$ (Disease-free):] $S_v\neq 0, F\left( S_v,P\right) = 0$ and $P=0, G\left(S_v,P \right) \neq 0$,\\
	\item[$E_3$ (Endemic):] $S_v\neq 0, F\left( S_v,P\right) = 0$ and $P\neq0, G\left(S_v,P \right) = 0$,\\
	\item[$E_4$ (Vector-free):] $S_v = 0, F\left( S_v,P\right) \neq 0$ and $P\neq0, G\left(S_v,P \right) = 0$.
\end{description}
Although there exists equilibrium points at first three instances $E_1,E_2, \text{ and } E_3$, a vector-free equilibrium, $E_4$ does not exist in the system \eqref{holling-specific}-\eqref{final model} (see Theorem \ref{th:vf}).\\

\begin{theorem} \label{th:vf}
There does not exist a vector-free equilibrium point in the system \eqref{holling-specific}-\eqref{final model} $\forall \left( S_v^*,P^*\right)  \in \mathbb{R}_+^2$.
\end{theorem}
\begin{proof}
Suppose the system in \eqref{holling-specific}-\eqref{final model} has an equilibrium point $E_4 = \left( S_v^*,P^*\right)$ where $S_v^*=0$ and $P^*\neq 0$. Since $E_4$ is an equilibrium point, by Definition \ref{def} we have $P^*=0$, i.e.
\begin{align*}
	P^* G\left(S_v^*,P^* \right) &=P^* \left( f_p\left(S_v^* \right)  - \delta\right) =0\\
\text{Since }P^*\neq 0 &\implies	f_p\left(S_v^* \right)  - \delta = 0 \\
\text{Since } S_v^* = 0 &\implies -\delta = 0
\end{align*}
However, $\delta>0$ by definitions of the system. Therefore this is a contradiction.\\
\noindent
Hence, the system \eqref{holling-specific}-\eqref{final model} does not contain vector-free equilibrium point.
\end{proof}
\noindent By Theorem \ref{th:vf}, we validate the model structure by demonstrating that mathematically impossible scenarios are excluded under our formulation, thereby confirming the internal consistency and robustness of the model. The equilibrium points for the system are given in Table \ref{Table:eq points}. 
\begin{table}[ht!]
	\caption{ Equilibrium points of the model in \eqref{holling-specific} - \eqref{final model} with functional responses I, II and III. }
	\centering
	\begin{tabular}{lcc}
		Equilibrium type& 	\begin{tabular}[c]{@{}l@{}}Functional \\ Response\end{tabular} & Equilibrium point $\left( S_v^*,P^*\right) $ \\ \toprule
		\vspace{0.3cm}	
		Trivial ($E_1$)         &I, II, III     &   $\left( 0,0\right)$ \\
		\vspace{0.3cm}   
		Disease-free ($E_2$)    & I, II, III       & $\left( v_cK,0\right) $     \\ 
		\vspace{0.3cm}   
		\multirow{3}{*}{Endemic ($E_3$)}& I& $\left(\dfrac{\delta}{\alpha \beta_v v_c} , \dfrac{\mu_v}{\beta_v}\left(1-\dfrac{\delta}{\alpha \beta_v v_c^{2}K}\right)\right) $   \\  
		\vspace{0.3cm}
		&II& $\left(\dfrac{\delta}{\beta_v v_c (\alpha-\delta E_p )},\dfrac{\alpha \mu_v\left( K\beta_v v_c^2 (\alpha-\delta E_p)- \delta\right) }{K\beta_v^2v_c^2(\alpha-\delta E_p )^2} \right)$ \\
		\vspace{0.3cm}
		&III$^*$&$\left( \dfrac{\sqrt{\delta}}{\sqrt{\beta_v v_c (\alpha-\delta E_p )}}, \dfrac{\alpha \mu_v\left( Kv_c\sqrt{\beta_vv_c(\alpha-\delta E_p)}-\sqrt{\delta}\right) }{K\beta_v v_c (\alpha-\delta E_p)\sqrt{\delta}}\right) $ \\
		\vspace{0.3cm}
		\multirow{1}{*}{Vector-free ($E_4$)}& I,II, III& Does not exist (Theorem \ref{th:vf}).\\  
		\bottomrule
	\end{tabular}
	\parbox{\textwidth}{* Type III generates an alternative endemic equilibrium; however, it is not included due to its ecological infeasibility.}
	\label{Table:eq points}
\end{table}


The stability analysis of the system incorporating the Holling type II functional response is presented in this section. To assess the local stability of equilibrium points, the Jacobian matrix of the system is derived and analyzed. Specifically, the signs of the determinant and trace of the Jacobian are used to determine the nature of the equilibria (see theorem \ref{thm: stability}). The corresponding stability analyses for systems with Holling type I and type III functional responses are provided in the supplementary material. Accordingly, the Jacobian matrix of the model with the functional response II is:\\

$$J=\begin{bmatrix}
v_c\mu_v\left(1-\dfrac{2S_v}{v_c K} \right)-\dfrac{\beta_v v_c P}{(1+\beta_v v_c E_p S_v)^2} & -\dfrac{\beta_v v_c S_v}{1+\beta_v v_c E_p S_v}\\
\dfrac{\alpha P}{(1+ \beta_vv_c E_p S_v)^2} & \dfrac{\alpha \beta_v v_c S_v}{1+ \beta_vv_c E_p S_v}-\delta
\end{bmatrix}$$

\begin{theorem} \label{thm: stability}
Consider the 2D nonlinear system
\[
\frac{dx}{dt} = f(x),
\]
where $f:\mathds{R}^2\rightarrow \mathds{R}^2$ is a smooth mapping. Let $x^*$ be an equilibrium point, i.e. $f(x^*)=0$. The equilibrium at $x^*$ is locally asymptotically stable if and only if the Jacobian matrix: $J=D f(x^*)$ has both eigenvalues with negative real parts, which is equivalent to, $\operatorname{tr}(J)<0 \quad \text{and} \quad \det(J)>0$.
\end{theorem}

Accordingly, the stability conditions of the three equilibrium points of the system is summarized in Table \ref{eql}.


\begin{table}[ht!]
\centering
 \caption{\label{eql} The stability conditions of the equilibrium points with trace determinant respectively for the model in \eqref{holling-specific} - \eqref{final model} for $\delta, v_c, \beta_v>0$ with functional response II}
 \begin{tabular}{ c l}
Equilibrium& tr$(J) <0$ and $\det(J) >0$\\
\toprule
$E_1\left( 0,0\right) $	&	$v_c<0$ (violates the domain of $v_c$)\\
$E_2\left(S_v^*,0 \right) $	&	$0<v_c<\dfrac{\delta}{\beta_v K (\alpha-\delta E_p)}$\\ 
 $E_3\left( S_v^*,P^*\right) $	&	$\dfrac{\delta}{K\beta_v(\alpha-\delta E_p)}<v_c<\dfrac{\alpha+\delta E_p}{K\beta_vE_p(\alpha-\delta E_p)}$\\
 \bottomrule
 \end{tabular}

\end{table}
	
\noindent	The model system in \eqref{holling-specific} - \eqref{final model} is given below specifically for functional response II.
	\begin{subequations}
		\begin{align}
			\dfrac{dS_v}{dt}&=v_c \mu_vS_v\left(1-\dfrac{S_v}{v_cK}\right)-\dfrac{\beta_v v_cS_vP}{1+\beta_v v_cE_pS_v} \label{eq:2_a}\\
			\dfrac{dP}{dt}&=\dfrac{\alpha\beta_v v_cS_vP}{1+\beta_v v_cE_pS_v} - \delta P \label{eq:2_b}
		\end{align}
	\end{subequations}
	
	\vspace{0.25cm}
	\begin{definition}
	An equilibrium point ${x}^*$ of a dynamical system is \textit{asymptotically stable} if for every solution ${x}(t)$ starting sufficiently close to ${x}^*$, we have
\[
\lim_{t \to \infty} {x}(t) = {x}^*.
\]
	\end{definition}
	
	\vspace{0.25cm}
	
	\begin{definition}
		An equilibrium point of a system is globally stable if it is stable for almost all initial conditions, not just those that are close to it.
	\end{definition}
	
	\vspace{0.25cm}
	
	\begin{theorem}
		Disease-free equilibrium of system in \eqref{eq:2_a}-\eqref{eq:2_b} is globally stable when $v_c^2< \dfrac{ \delta}{\beta_v K(\alpha-\delta E_p)}$. 
	\end{theorem}
	
	\begin{proof}
 Given that $\mu_v, v_c,\beta_v >0$ and $\dfrac{dS_v}{dt} <0$ whenever $S_v > v_c K$, the solution $S_v(t)$ is a decreasing whenever $S_v(0)> v_c K$. 
		 
		\noindent Assume for some $t_0>0$, $S_v(t_0)=v_c K$. Then $\dfrac{dS_v}{dt}(t_0) \leq 0$ and $S_v \leq v_c K$ for all $t>t_0$. Therefore, for increasing function $f_1(S_v)$, 
		\begin{equation}
			\alpha f_1(S_v)-\delta = \dfrac{\alpha\beta_v v_cS_v}{1+\beta_v v_cE_pS_v} - \delta  \leq \dfrac{\alpha\beta_v v_c^2 K}{1+\beta_v v_c^2E_pK} - \delta =\delta \left(\dfrac{\alpha\beta_v v_c^2 K}{\delta(1+\beta_v v_c^2E_pK)} -1 \right) \label{inequaity}
		\end{equation}
		$\forall\, t> t_0$. The inequality in \eqref{inequaity} and Eq. \eqref{eq:2_b} yields, 
		\begin{align*}
			\dfrac{dP}{dt} \leq \delta \left(\dfrac{\alpha\beta_v v_c^2 K}{\delta(1+\beta_v v_c^2E_pK)} -1 \right)P 
		\end{align*}
		and consequently,
		\begin{equation*}
			0\leq P \leq c_1 \exp \left({\delta\left(\frac{\alpha\beta_v v_c^2 K}{\delta(1+\beta_v v_c^2E_pK)} -1 \right)(t-t_0)}\right)
		\end{equation*}
		for all $t>t_0$ for some $c_1 \in \mathbb{R}$.
		Suppose $v_c^2< \dfrac{ \delta}{\beta_v K(\alpha-\delta E_p)}$. Then $\lim_{t \to \infty}P = 0$ for any initial condition of $P$. 
		
		\noindent Now let us assume $\exists\,$ another $t_0$ such that $S_v(t_0) <v_c K$. From equation \eqref{eq:2_a}, we have 
		\begin{equation*}
			\dfrac{dS_v}{dt} \leq v_c \mu_vS_v\left(1-\dfrac{S_v}{v_cK}\right) \leq v_c^2 K \mu_v \left(1-\dfrac{S_v}{v_cK}\right)
		\end{equation*}
		for $t \geq t_0$. Integrating the resulting inequality, 
		\begin{equation*}
			S_v(t) \leq v_c K+ c_2 e^{-v_c\mu_v (t-t_0)},
		\end{equation*} for $t\geq t_0$ for some $c_2 \in \mathbb{R}$. Hence $\limsup_{t \to \infty} S_v(t) \leq v_c K$. With $\lim_{t \to \infty} P(t)=0 $, for every $\epsilon>0$ there exists $t'>t_0$ such that $P(t) \leq \epsilon$ whenever $t>t'$. For an arbitrary $\epsilon$ we have,
		
		\begin{equation*}
			\dfrac{dS_v}{dt} \geq v_c \mu_vS_v\left(1-\dfrac{S_v}{v_cK}\right)-\dfrac{\beta_v v_cS_v\epsilon}{1+\beta_v v_cE_pS_v} \geq v_c \mu_vS_v\left(1-\dfrac{S_v}{v_cK}\right) - \beta_v v_cS_v\epsilon
		\end{equation*}
		for $t\geq t'$. Integrating the resulting inequality, 
		\begin{equation*}
			S_v(t) \geq \dfrac{v_c \mu_v - \epsilon \beta_v v_c}{\dfrac{\mu_v}{K}+c_3 e^{-(v_c\mu_v-\epsilon \beta_vv_c)(t-t')}}
		\end{equation*} 
		for $t \geq t'$ for some $c_3 \in \mathbb{R}$ and therefore,
		\begin{equation}
			\liminf_{t \to \infty} S_v(t) \geq \dfrac{Kv_c \mu_v -K\epsilon \beta_v v_c}{\mu_v} \label{inequality2}
		\end{equation}
		  for all $\dfrac{\mu_v}{\beta_v} >\epsilon > 0$. Since \eqref{inequality2} holds for all $\dfrac{\mu_v}{\beta_v} >\epsilon > 0$, $\liminf_{t \to \infty} S_v(t) \geq v_cK$.
		
		\noindent Hence $\lim_{t \to \infty} S_v$ exist from $\liminf_{t \to \infty}S_v \geq v_cK$ and $\limsup_{t \to \infty} S_v \leq v_cK$,
		\begin{equation*}
			\lim_{t \to \infty} S_v = v_cK
		\end{equation*}
		for $S_v(t_0) < v_cK$ for any $t_0 >0$.
		
		\noindent Suppose that there is no such $t_0$ where $S_v(t_0)=v_c K$. Since $\nexists$ such $t_0$ and $S_v(t)$ is a decreasing function, $S_v(t) > v_c K$. Then $S_v(t)$ must converges to either $v_c K$ or some $a\in \mathbb{R}$ s.t. $a> v_cK$. Suppose $\lim_{t \to \infty} S_v(t)=a$. Then $\lim_{t \to \infty} \dfrac{dS_v}{dt}=0$ and from RHS of \eqref{eq:2_a} we get,
		\begin{equation}
			\lim_{t \to \infty}\dfrac{\beta_v v_cS_vP}{1+\beta_v v_cE_pS_v}=\lim_{t \to \infty}v_c \mu_vS_v\left(1-\dfrac{S_v}{v_cK}\right)= v_c \mu_va\left(1-\dfrac{a}{v_cK}\right). \label{limit}
		\end{equation}
		However the limit in \eqref{limit} is negative for $a> v_c K$. This is a contradiction by theorem 3.1 ($S_v(t), P(t)>0$ for all positive initial conditions). Hence $S_v(t)$ does not converge to $a$ but definitely converges to $v_c K$. With $\lim_{t \to \infty} S_v(t) = v_c K$, for every $\epsilon>0$, $\exists\, t'$ s.t. $\left| S_v(t)-v_c K\right| < \epsilon$ whenever $t> t'$. For any arbitrary $\epsilon>0$ we have $S_v(t) < \epsilon + v_cK$ and,
		\begin{equation}
			\dfrac{dP}{dt} \leq \delta\left( \dfrac{\alpha\beta_v v_c(v_cK+\epsilon)}{\delta(1+\beta_v v_cE_p(v_cK+\epsilon))} - 1 \right) 
		\end{equation}
		for all $t> t'$ and consequently,
		\begin{equation}
			0 \leq P \leq c_4 \exp \left( \delta\left( \dfrac{\alpha\beta_v v_c(v_cK+\epsilon)}{\delta(1+\beta_v v_cE_p(v_cK+\epsilon))} - 1 \right)(t-t') \right) \label{inequality 3}
		\end{equation}
		for all $t>t'$ for some $c_4\in \mathbb{R}$. Since inequality \eqref{inequality 3} holds for all $\epsilon >0$, 
		\begin{equation}
			0\leq P \leq c_4 \exp \left({\delta\left(\frac{\alpha\beta_v v_c^2 K}{\delta(1+\beta_v v_c^2E_pK)} -1 \right)(t-t')}\right).
		\end{equation}
		Therefore $\lim_{t \to \infty}P = 0$ for any initial condition of $P$.
		
		Finally, we can conclude that for any initial condition, the trajectories $S_v(t)$ and $P(t)$ converge to $v_c K$ and $0$, respectively. Therefore, by the definition of global stability, disease-free equilibrium of the system in \eqref{eq:2_a}-\eqref{eq:2_b} is globally stable whenever $v_c^2< \dfrac{ \delta}{\beta_v K(\alpha-\delta E_p)}$.\
	\end{proof} 
	\noindent
	Global stability of the disease-free equilibrium points of the models with functional response I and III can be derived following the same logic. Thus we can conclude that disease-free equilibrium is globally stable for our model with their respective local conditions given in Table \ref{eql}. Although the endemic equilibrium is locally stable when $\dfrac{\delta}{K\beta_v(\alpha-\delta E_p)}<v_c<\dfrac{\alpha+\delta E_p}{K\beta_vE_p(\alpha-\delta E_p)}$, 
	making a statement about its globally stability is not direct. Nevertheless, endemic equilibrium being the only remaining locally stable equilibrium of the system, the likelihood of it being global stable is very high. However it must be highlighted, due to the nonlinear nature of the system, the existence of periodic orbits and chaotic dynamics of the endemic equilibrium cannot be ruled out without a thorough analysis.

\section{Numerical Results and Discussion}
\subsection{Endemic equilibrium and $v_c$}
In this section, we present the numerical results of the system in \eqref{holling-specific} - \eqref{final model}. In that we used vector mortality rate ($\mu_v$) to be $1/6$ \citep{VDCI}, transmission rate from host to vector ($\beta_v$) to be $0.375$ \citep{G_12}, and extrinsic incubation period ($E_p$) to be 2 \citep{G_96}. We show the stability of equilibrium points in phase portraits (see Fig. \ref{fig:E3_track}). We exclude $E_1$ from the discussion due to its lack of relevance and importance to the biological problem under investigation. For all three functional responses, the system undergoes a transcritical bifurcation (see Definition \ref{bif_def}) at which the stability of the interior equilibrium $E_3$ and the boundary equilibrium $E_2$ is exchanged. At this critical point, the two equilibria coincide, and their stability properties are reversed, signaling a fundamental qualitative shift in system dynamics. This transition reflects a change from a disease-free regime to an endemic regime, highlighting how subtle changes in parameters governing vector–pathogen interactions can fundamentally alter population-level disease outcomes. Notably, as $v_c$ increases beyond this threshold, the interior equilibrium $E_3$ exhibits a rapid and pronounced displacement toward the $P-$ axis, accompanied by a retreat from the $S_v-$ axis. This trajectory reveals that higher $v_c$ sharply amplifies pathogen replication by enabling infection of a larger fraction of the vector population. Consequently, even modest increases in $v_c$ beyond the bifurcation point can significantly elevate pathogen load, reinforcing endemic persistence and demonstrating the system’s heightened sensitivity to $v_c$. \\

For a Type-I functional response ($q = 0$), endemicity emerges once $v_c$ exceeds $v_c = 0.3651$, corresponding to a minimum transmission efficiency of 36.51\%. Increasing the nonlinearity to Type II ($q = 1$) raises this threshold to 38.49\%, while a Type-III response requires 52.91\% efficiency to maintain endemic levels. Notably, the transition from Type II to III nearly doubles the transmission efficiency required, highlighting how nonlinear vector–pathogen interactions can substantially elevate the barrier to disease persistence. Functional response III represents intensified competition for resources arising from an increased number of susceptible vectors under favorable environmental conditions. As shown in Fig.~\ref{fig:E3_track}, this scenario leads to a stable endemic state sustained by a relatively small vector population (evident from the displacement of the \(E_3\) trajectory along the x-axis), yet it generates a higher pathogen load compared to the Type~I and Type~II cases. Consequently, when vector abundance rises following favorable weather, the infection intensity can escalate even with a smaller active vector pool. Critically, if resource availability---such as human hosts---also increases proportionally with vector proliferation, the endemic intensity may amplify dramatically, heightening the risk of large-scale outbreaks.\\

\begin{definition}\label{bif_def}
Consider a dynamical system depending smoothly on a parameter $\lambda \in \mathbb{R}$. Suppose that at a critical value $\lambda = \lambda_0$ two equilibrium branches intersect. 
If, as $\lambda$ passes through $\lambda_0$, the stability of these equilibria is exchanged, then the system undergoes a transcritical bifurcation at $\lambda = \lambda_0$. In this case, one equilibrium branch is stable and the other unstable for $\lambda < \lambda_0$, while for $\lambda > \lambda_0$ their stability is reversed.\\
\end{definition}

 \begin{figure}[ht!]
 	\centering
 	\includegraphics[width=\linewidth]{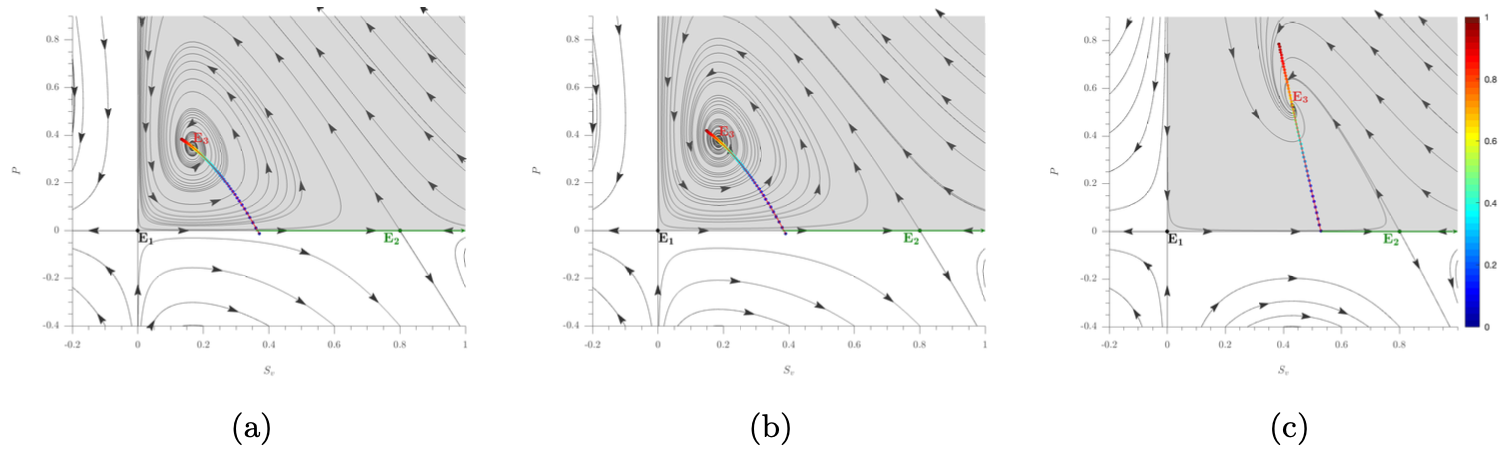}
 	\caption{\label{fig:E3_track} The stability of $E_3$ is illustrated for functional responses (a) Type I ($q=0$), (b) Type II ($q=1$), and (c) Type III ($q=2$) in model \eqref{holling-specific} - \eqref{final model}. The trajectories of $E_3$ and $E_2$ are traced as $v_c$ increases beyond their respective transcritical bifurcation thresholds; $v_c = 36.51\%,38.49\%,\text{ and }52.91\%$ respectively. Values for other parameters were set at $\mu_v=0.166, \beta_v=0.375, E_p=2, K=1, \delta=0.1$, $\alpha=2$ (references of the literature are given in Table \ref{Table: parameter definition}).}
 \end{figure}
 
  \begin{figure}[ht!]
 	\centering
 	\includegraphics[width=0.5\linewidth]{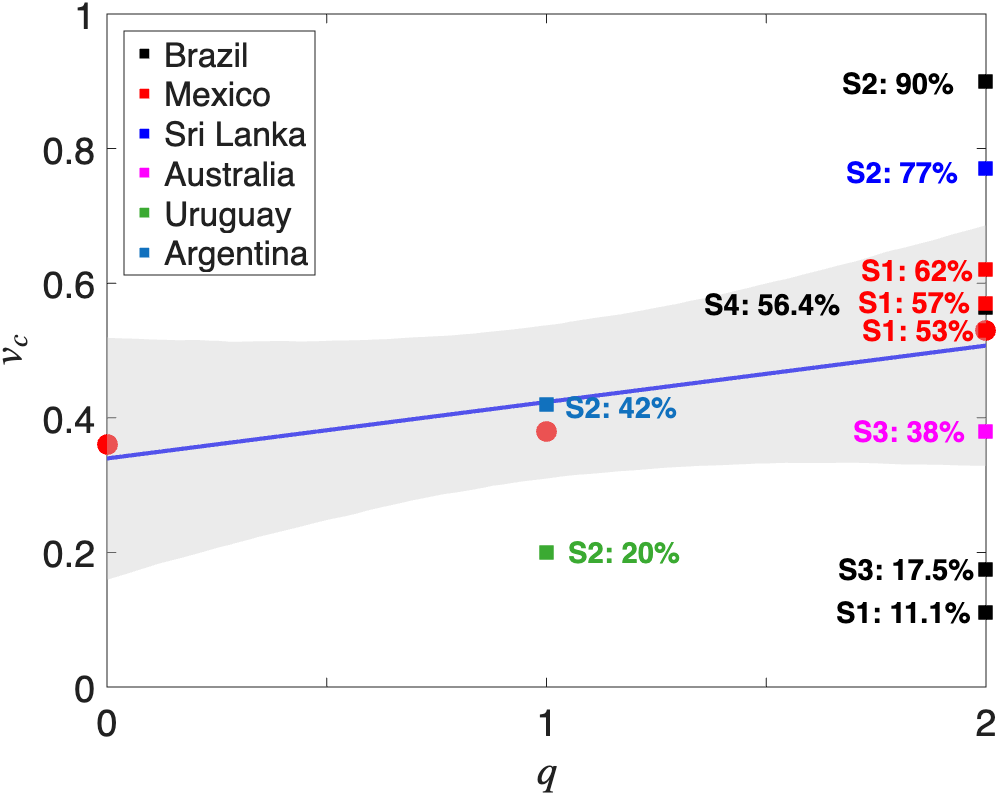}
 	\caption{\label{fig:q_vc} Variation in the vector competence ($v_c$) threshold required for endemicity across Holling functional response types, illustrating how resource-dependent consumption rates elevate the competency level needed for a stable endemic state. The red points are the $v_c$ thresholds from bifurcation and the blue solid like is a linear fit ($v_c = 0.3382 + 0.0841q$) with shaded area depicting uncertainty quantification (95\% credible interval). The $v_c$ with their serotype are marked at $q=1$ for sub-tropical countries and $q=2$ for tropical countries, where dengue is prevalent.}
 \end{figure}

Figure \ref{simulations_compare} illustrates the temporal dynamics of the two populations under functional responses I–III. The results show that as the functional response transitions from Type I to Type III, the resource abundance increases while consumers' top-down pressure increases and the system reaches equilibrium more rapidly. This accelerated convergence reflects the vector population’s adaptive response to environmental or biological stressors—under resource limitation, mosquitoes swiftly adjust to coexist with these pressures, balancing their interaction with pathogens and thereby hastening endemic establishment \citep{kalinkat2023empirical}. These results underscore that resource scarcity among abundant susceptible vectors can intensify infection levels. Paradoxically, climatic conditions such as rainy seasons— favorable for mosquito proliferation—may further enhance viral transmission capacity by accelerating adaptive responses within the vector population.\\

\begin{figure}[ht!]
\centering
\begin{subfigure}{0.3\linewidth}
			\includegraphics[width=\linewidth]{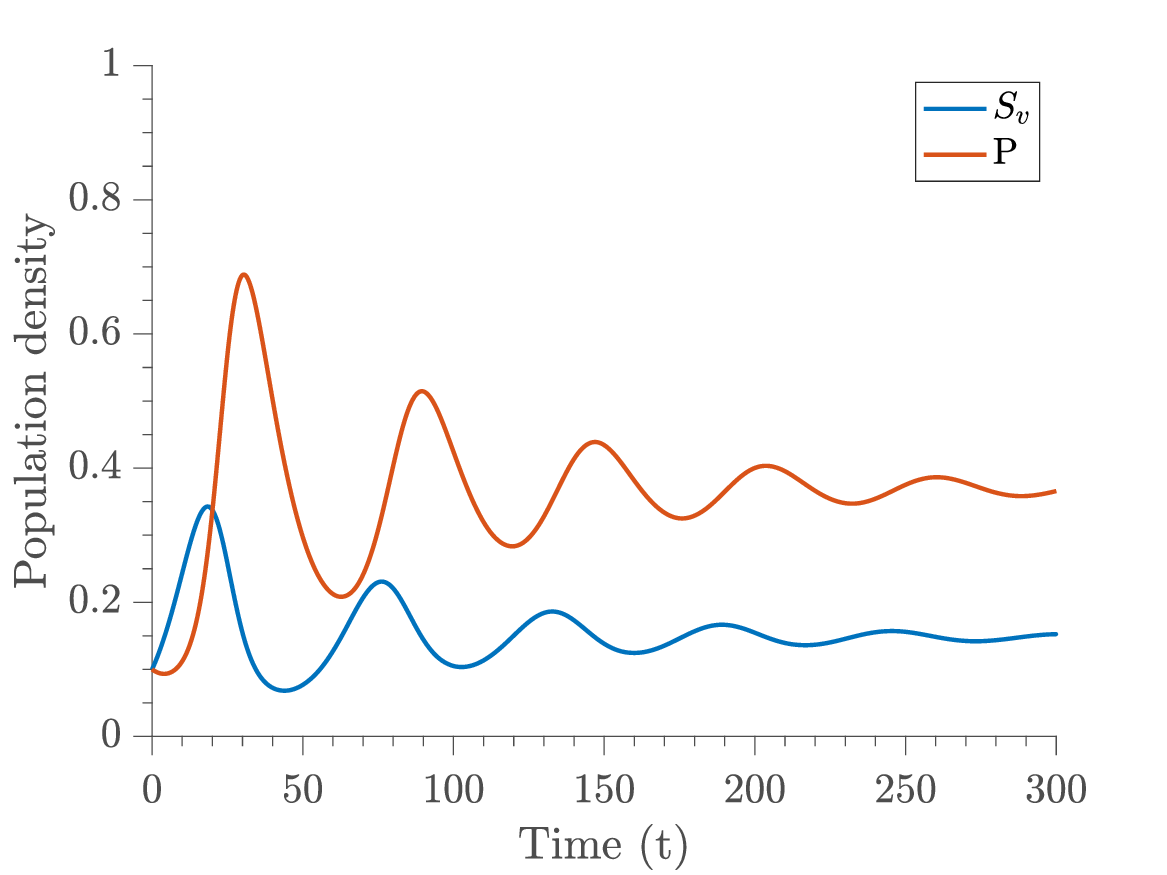}
			\caption{}
		\end{subfigure}
		\begin{subfigure}{0.3\linewidth}
			\includegraphics[width=\linewidth]{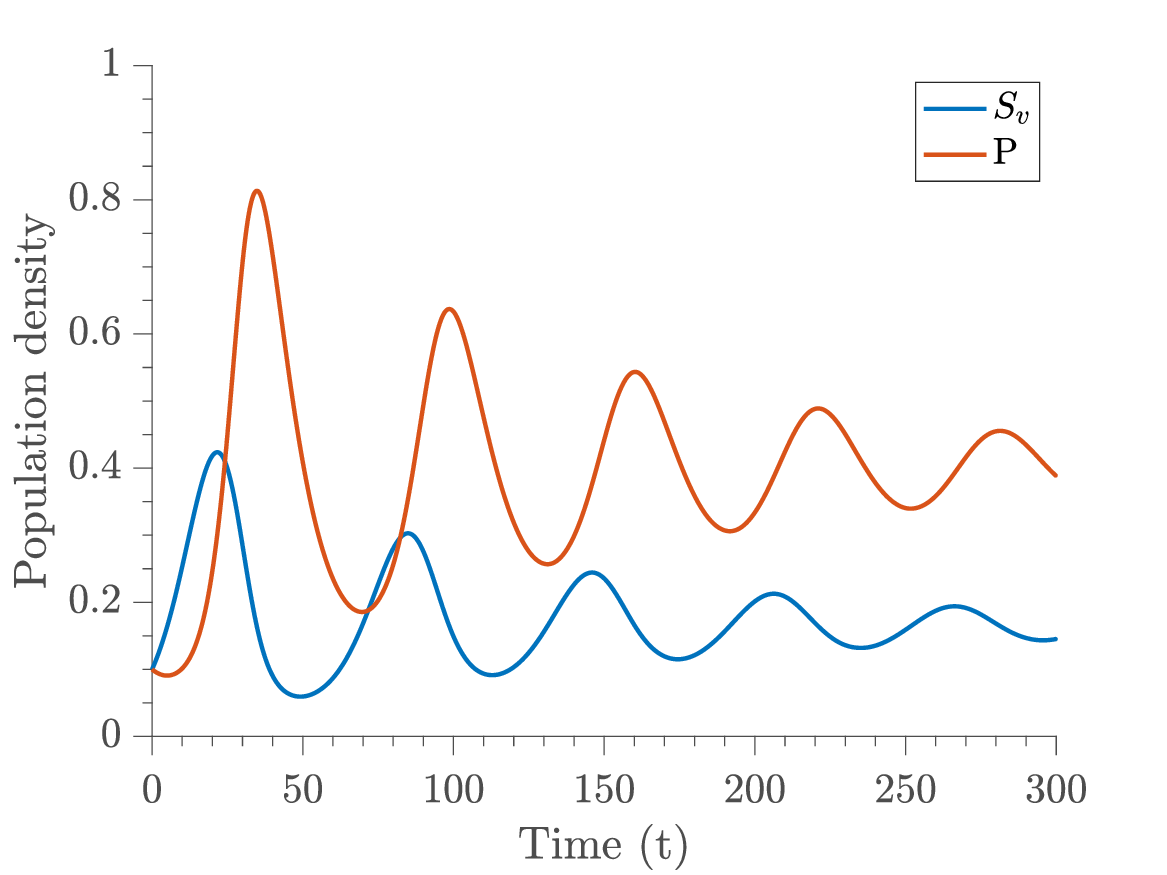}
			\caption{}
		\end{subfigure}
		\begin{subfigure}{0.3\linewidth}
			\includegraphics[width=\linewidth]{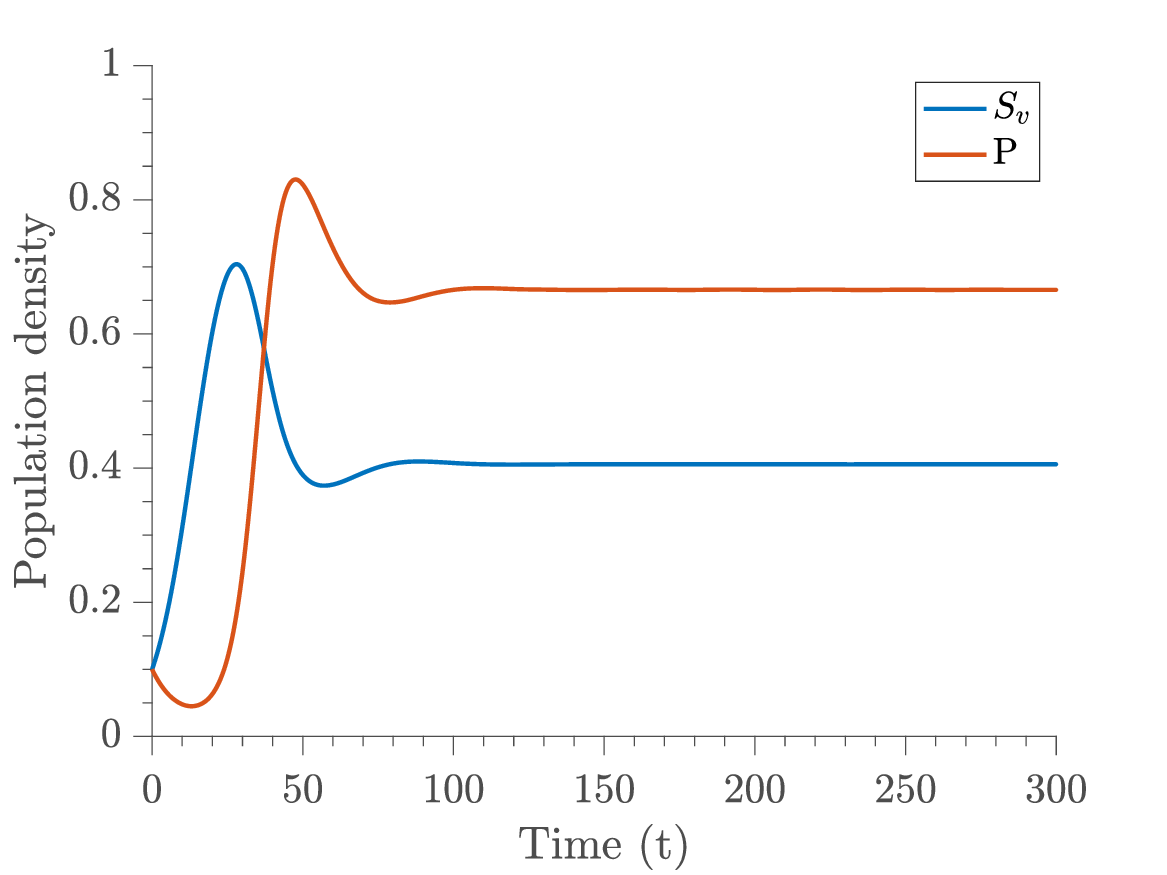}
			\caption{}
		\end{subfigure}	
	\caption{\label{simulations_compare} Comparison of simulations of the system \eqref{holling-specific}-\eqref{final model}. Model simulations for the system for (a) $q=0$, (b) $q=1$, and (c) $q=2$. The parameter values used are $v_c=0.9, \mu_v=0.166, \beta_v=0.375, E_p=2, K=1, \delta=0.1$ and $\alpha=2$. }
\end{figure}

\subsection{Pathogen inactivation rate and $v_c$}
As illustrated in Fig.~\ref{FR_stability_change}, the two-parameter bifurcation diagram delineates the combinations of $\delta$ and $v_c$ required to sustain an endemic state. It marks the critical $(v_c, \delta)$ points where system stability transitions from the disease-free equilibrium ($E_2$), represented by the green region, to the endemic equilibrium ($E_3$), represented by the red region, thereby underscoring the coupled influence of pathogen inactivation and $v_c$ in maintaining endemicity. \\

Vectors possess an innate immune system that lacks adaptive capability, implying that their maximum energetic investment in immune responses to pathogen pressure is inherently limited. Consequently, the only viable evolutionary pathway available to them lies in modulating $v_c$, which governs the efficiency of pathogen acquisition and transmission. Therefore, any evolutionary adaptation to environmental or pathogen pressures must occur through optimizing $v_c$. The outcomes of this adaptation are illustrated in Fig.~\ref{fig:E3_track}, where the trajectory of $E_3$ shifts markedly across the $S_v$–$P$ plane as $v_c$ increases, reflecting enhanced pathogen persistence. Since there is limited empirical evidence on feasible values of the pathogen inactivation rate ($\delta$), it remains uncertain whether vectors in a given region can achieve the levels of $\delta$ required to counterbalance rising $v_c$, potentially leading to outbreaks. This constraint highlights a fundamental biological trade-off: vectors can evolve higher competence but not stronger immune capacity, thereby defining the limits of their evolutionary response to pathogen-driven selection and contributing to instability in disease transmission dynamics.

\begin{figure}[ht!]
	\centering
				\includegraphics[width=0.7\textwidth]{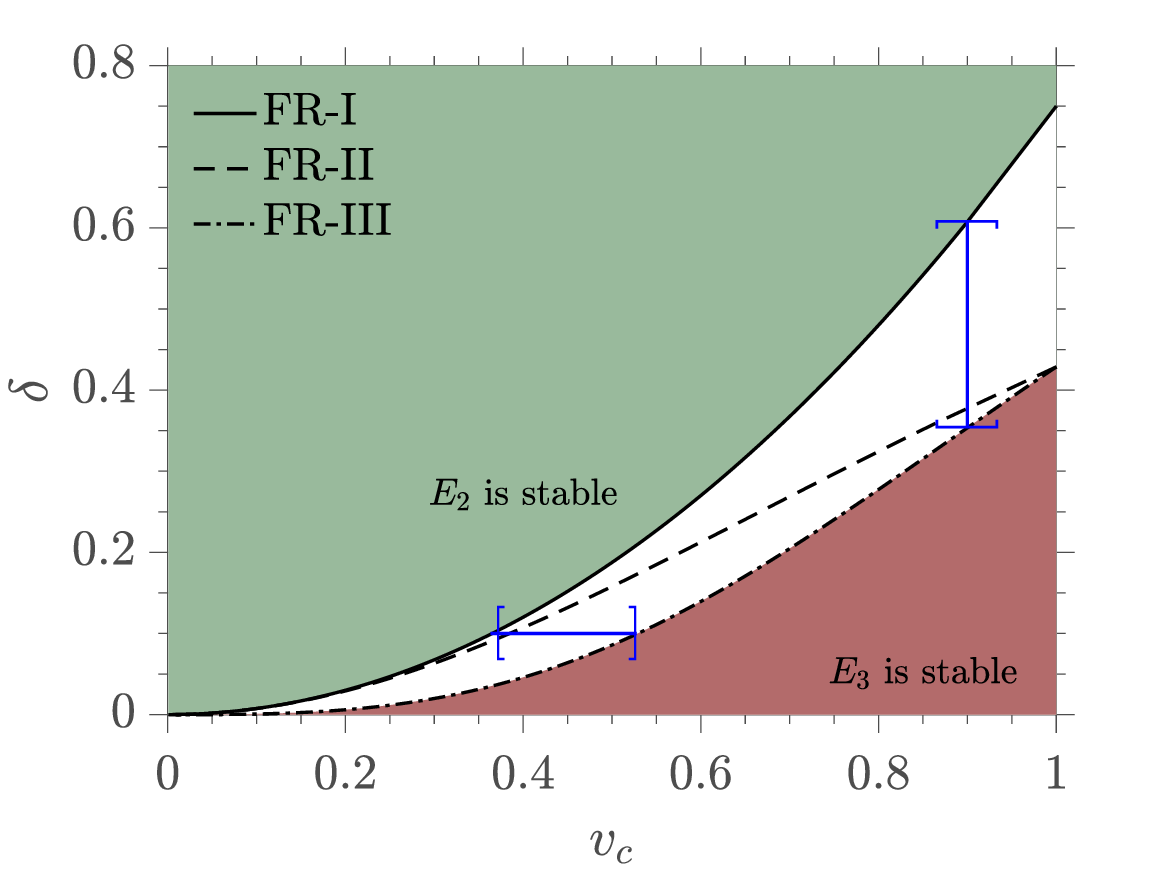}
	\caption{\label{FR_stability_change} Two parameter bifurcation diagram: The range of $\delta$ and $v_c$ for which the equilibrium points $E_2$ (disease-free) and $E_3$ (endemic) are stable for all functional responses. Red shows the endemic region and green shows the disease-free region. The white area shows the uncertainty of choosing between functional responses. }
\end{figure}

The white region at which $E_2$ changes to $E_3$ explains the uncertainty of functional response choice - exponential growth rate to resource-limited growth rate triggering interspecies competition. For a fixed pathogen inactivation rate of $\delta = 0.1$, the transition from the disease-free to the endemic state is highly sensitive, with the corresponding uncertainty in $v_c$ ranging from $36.51\%$ to $52.91\%$. Conversely, when $v_c$ is fixed at $0.9$, the uncertainty in $\delta$ lies within the interval $\left(35.43\%, 60.80\%\right)$. These ranges reveal that small variations in either parameter can decisively alter system stability, indicating that even marginal changes in vector immunity or transmission efficiency may trigger or suppress endemic persistence.\\



From literature studies conducted where dengue is prevalent, we extracted $v_c$ values corresponding to different regions as well as different serotypes (see Fig. \ref{fig:q_vc}). The summary of these $v_c$ values can be found in \citep{chathurangika2024determining} and references therein. Based on these values comparing against the two-parameter bifurcation diagram in Fig. \ref{FR_stability_change}, we argue the level of $\delta$, assuming an endemic state. For instance, in Brazil, a tropical country characterized by frequent dengue outbreaks, $v_c$ exhibits a clear increase across serotypes. The persistent endemicity observed at these $v_c$ levels suggests a critical biological constraint: mosquitoes are unable to allocate more than approximately 38\% of their resources toward immune defense. This limitation implies that further increases in $v_c$ cannot be counterbalanced by enhanced immunity, effectively setting an upper bound on pathogen suppression within the vector population. Consequently, the interplay between high transmission efficiency and constrained immune investment drives sustained endemic transmission, highlighting a fundamental evolutionary and epidemiological trade-off that shapes disease dynamics in highly endemic regions. The difference between the immune investment required to counterbalance $v_c$ diminishes as $v_c$ increases, indicating that vectors have progressively less capacity to suppress pathogen replication at higher transmission efficiencies. This pattern underscores a fundamental biological constraint: under strong environmental pressure (e.g. tropical weather, monsoons, high population densities), even modest increases in $v_c$ can overwhelm the vector’s fixed immune defenses, dramatically elevating the likelihood of endemic transmission. Collectively, these results highlight a critical evolutionary and epidemiological trade-off—vectors can enhance transmission potential but cannot proportionally bolster immunity, making highly endemic conditions almost inevitable in favorable environments.\\

\section{Conclusion and future work}

In this study, we model the vector-pathogen interaction using a classical predator-prey framework. By incorporating Holling’s functional responses (Type I, II, and III), we analyzed how varying consumption rates influence the balance between vector tolerance and resistance. This approach allowed us to systematically investigate pathogen successful replication in \textit{Aedes} vectors and its contribution to virus transmission. Our results highlight vector competence ($v_c$)—the ability to acquire, replicate, and transmit the pathogen—as a key determinant of disease persistence\\

Through analytical investigation, we demonstrated the critical role of vector biology in shaping dengue persistence and control. By framing vector–pathogen interactions as a predator–prey system, we identified the range of $v_c$—an evolving trait governing pathogen acquisition, replication, and transmission—that determines whether dengue persists or is eradicated. Equilibrium analysis of the model in \eqref{holling-specific}–\eqref{final model} showed that as the functional response shifts from Type I to III, the lower bound of $v_c$ for disease persistence increases, reflecting stronger pathogen pressure under tropical and subtropical environmental conditions. The pathogen inactivation rate ($\delta$), representing the innate immune capacity of vectors constrain evolutionary responses: vectors can increase competence but cannot enhance immune defense, highlighting a fundamental biological trade-off.\\

This study focuses on vector–pathogen interactions to examine how $v_c$—an evolving trait—and the innate immune system shape dengue persistence under tropical and subtropical environmental pressures. While excluding explicit human hosts limits the model’s ability to capture human heterogeneity, immunity, and behavior, this simplification is intentional: the study aims to isolate the vector–pathogen subsystem to clarify fundamental biological trade-offs, rather than model full transmission dynamics. Environmental effects, such as seasonality and climate, are represented indirectly via vector carrying capacity and functional responses, and continuous pathogen exposure is assumed to simplify analysis of endemic conditions. Despite these constraints, the framework provides new insights into how $v_c$ evolution and fixed immune capacity determine disease persistence, highlighting the limits of vector adaptation and the mechanisms driving endemic stability. Extending the model to include human hosts, explicit climate variables, or variable pathogen abundances could enhance predictive accuracy and inform control strategies, but such extensions are beyond the scope of this work. Experimental and field studies could further validate these findings, clarifying the relative contributions of $v_c$ and immune defenses to disease dynamics.


\pagebreak

\bibliographystyle{plainnat}
\bibliography{references}


\section*{Supplementary information}

\begin{appendices}
	
	\section{Local Stability Analysis}\label{secA1a}
\subsection{Functional response Type - I}
The Jacobian matrix of the system in \eqref{holling-specific} - \eqref{final model}, when $q=0$ (Functional response I is given by,
$J_{q=0}=\begin{bmatrix}
\mu_vv_c\left(1-\dfrac{S_v}{Kv_c}\right)-P\beta_vv_c-\dfrac{S_v\mu_v}{K}		&	-S_v\beta_vv_c\\
P\alpha \beta_vv_c		&		S_v\alpha\beta_vv_c-\delta
\end{bmatrix}$

The Trace and determinant conditions for the stability of Equilibrium points of the system are given by,

\begin{table}[ht!]
\centering
\footnotesize
 \caption{\label{eql} The stability conditions of the equilibrium points of the model in \eqref{holling-specific} - \eqref{final model}}
 \begin{tabular}{c c c c}
\toprule
Functional Response	&Equilibrium& Tr $J $(stability condition)&$\det J $(stability condition)\\
\midrule
\multirow{3}{*}{I} &$E_1$ &	$v_c<\dfrac{\delta}{\mu_v}$	&	$-v_c\mu_v\delta>0$\\
&	$E_2$	&	$K\alpha \beta_v v_c^2-\mu_v v_c-\delta<0$	&	$v_c^2<\dfrac{\delta}{K\alpha \beta_v}$\\
&	$E_3$	&	$\dfrac{-\delta\mu_v}{K\alpha \beta_v v_c}<0$	&	$v_c^2>\dfrac{\delta}{\alpha K \beta_v}$\\
 \bottomrule
 \end{tabular}

\end{table}

\subsection{Functional response Type - III}

The Jacobian matrix of the system in \eqref{holling-specific} - \eqref{final model}, when $q=2$ (Functional response III is given by,

\vspace{0.25cm}
$J_{q=2}=\begin{bmatrix}
-\dfrac{2PS_v\beta_vv_c}{(S_v^2\beta_vE_pv_c+1)^2}+\mu_vv_c-2\dfrac{\mu_vS_v}{K}		&	\dfrac{-S_v^2\beta_vv_c}{S_v^2\beta_vE_pv_c+1}\\
\dfrac{2PS_v\alpha\beta_vv_c}{S_v^2\beta_vE_pv_c+1)^2}				&		\dfrac{S_v^2\alpha\beta_vv_c}{S_v^2\beta_vE_pv_c+1}-\delta
\end{bmatrix}$
\vspace{0.25cm}

The Trace and determinant conditions for the stability of Equilibrium points of the system are given by,

\begin{table}[ht!]
\centering
\footnotesize
 \caption{\label{eql} The stability conditions of the equilibrium points of the model in \eqref{holling-specific} - \eqref{final model}}
 \begin{tabular}{c c c c}
\toprule
Functional Response	&Equilibrium& Tr $J $(stability condition)&$\det J $(stability condition)\\
\midrule
 \multirow{3}{*}{III} &$E_1$ &	$v_c<\dfrac{\delta}{\mu_v}$	&	$-v_c\mu_v\delta>0$\\
 &$E_2$	&	$K^2\beta_v\mu_vhv_c^3-K^2\beta_v(\alpha-\delta E_p)v_c^2+\mu_v>0$& $v_c^3<\dfrac{\delta}{K^2\beta_v(\alpha-\delta E_p)}$\\
 &$E_3$	&	$v_c^3<\dfrac{2\delta^2 E_p}{K^2\beta_v(\alpha-\delta E_p)(2\delta E_p-\alpha)^2}$	&	$v_c^3>\dfrac{\delta}{K^2\beta_v(\alpha-\delta E_p)}$\\
 \bottomrule
 \end{tabular}

\end{table}

\end{appendices}

%
%
%
%
%
%
%
%
%
%
%
%
%
%
%
%
\subsection*{Author contribution}
P. C. - Investigation, Methodology, Numerical simulations, Writing the original draft, and editing. T. P. - Methodology, Writing Appendix. L. S. P. - Biological validation of the model and the results. S. S. N. P. - Review and discussions. K. D. S. - Conceptualization, Investigation, Methodology, Writing the original draft, and editing. All authors read and approved the final manuscript.

\pagebreak
\end{document}